%% file: main.tex
\documentclass[11pt]{article}

\usepackage[utf8]{inputenc}
\usepackage[T1]{fontenc}
\usepackage{lmodern}
\usepackage[margin=1in]{geometry}
\usepackage{amsmath,amssymb,amsthm}
\usepackage{mathtools}
\usepackage{bm}
\usepackage{graphicx}
\usepackage{booktabs}
\usepackage{microtype}
\usepackage{xcolor}
\usepackage{enumitem}
\usepackage{caption}
\usepackage{subcaption}
\usepackage[numbers,sort&compress]{natbib}
\usepackage[colorlinks=true,linkcolor=blue!55!black,citecolor=teal!60!black,urlcolor=blue!55!black]{hyperref}
\usepackage{cleveref}

\newtheorem{proposition}{Proposition}

\newcommand{\R}{\mathbb{R}}
\newcommand{\E}{\mathbb{E}}
\newcommand{\conv}{\operatorname{conv}}
\newcommand{\sep}{\dagger}
\newcommand{\Down}{\mathrm{Down}}
\newcommand{\Up}{\mathrm{Up}}

\newcommand{\Dsep}{d_{\mathrm{sep}}}
\newcommand{\Mset}{\mathcal{M}}
\newcommand{\Pset}{\mathcal{P}}
\newcommand{\dom}{D(E)}

\newcommand{\indic}[1]{\mathbf{1}\!\left[#1\right]}
\DeclareMathOperator*{\argmax}{arg\,max}

\DeclareMathOperator*{\softmax}{softmax}

\input{results_macros}

\title{\bfseries Rendering Separoid Information:\\[2pt]
Rate--Distortion Reconstruction of Convex Apartness Scenes}

\author{%
  Faruk Alpay\thanks{Corresponding author: \texttt{alpay@lightcap.ai}}\quad Bar\i{}\c{s} Ba\c{s}aran\\[2pt]
  \normalsize Department of Computer Engineering, Bah\c{c}e\c{s}ehir University\\
  \normalsize Istanbul, Turkey\\
  \normalsize \texttt{\{faruk.alpay,\,baris.basaran\}@bahcesehir.edu.tr}
}
\date{}

\begin{document}
\maketitle

\begin{abstract}
A convex scene communicates more than shape: the pattern of which groups of
objects are mutually \emph{apart} and which \emph{cross} is a discrete, logical
payload that a picture must carry to a viewer or a downstream decoder. We make
this payload explicit. The apartness table of a finite family of convex
bodies, a \emph{separoid}, is treated as a source signal, a renderable convex
scene as its encoder, and the rendered image as a noisy visual channel from
which the apartness structure is decoded. Within this view we pose
\emph{apartness-preserving rendering as a rate--distortion problem}: minimise a
geometry description length subject to faithfully transmitting the separation
\emph{certificates} (the maximal separations and minimal Radon partitions that
generate the table by closure) rather than minimising pixel error. The
distortion is closure-aware (a high-consequence certificate bit costs more than
a redundant one), the rate is a differentiable geometric code length, and
apartness is encoded as a distribution over separating directions whose breadth
sets its robustness to viewpoint. We give a differentiable support-function
realisation of the pipeline, a variational lower bound on the apartness mutual
information $I(\Sigma;Y)$, and an information-theoretic account of view
selection. Experiments on planar convex scenes show that (i) scenes are
recovered from the apartness table alone at \reconHamSix{}\% bit accuracy, with
the certificate skeleton alone already determining the full table; (ii) a clean
operational rate--distortion frontier emerges, along which the
consequence-weighted certificate distortion exceeds plain Hamming error, since
low-rate rendering preferentially damages the high-consequence structure;
(iii) a rendered
$48\times48$ image delivers \channelISfortyeight{} of the apartness-graph
entropy, with the information score rising with resolution and collapsing with
noise like a capacity-limited channel; and (iv) information-optimal views
certify most separations from a single direction, and a small rate surcharge
makes them robust across a wide cone of viewpoints.
\end{abstract}

\section{Introduction}
Rendering is usually evaluated by how closely an image reproduces a reference
appearance. Yet many figures exist to communicate a \emph{relational} fact:
that these objects form a cluster while those stand apart, that two groups
overlap, that a configuration cannot be pulled apart by any straight cut. In a
diagram of a mechanism, a molecular cartoon, a scene graph, or a statistical
embedding, the geometry is only a carrier; the message is the pattern of
separations and crossings. When that message is what matters, the right
objective is not pixel fidelity but \emph{how many bits of relational structure
the picture transmits, and at what geometric cost.}

We formalise this for convex scenes. A finite family of convex bodies
$C_1,\dots,C_n\subset\R^d$ induces, for every pair of disjoint index sets
$A,B$, a single bit
\[
  \sigma_{A,B}=1 \iff \conv\!\Big(\textstyle\bigcup_{a\in A}C_a\Big)\cap
  \conv\!\Big(\textstyle\bigcup_{b\in B}C_b\Big)=\varnothing .
\]
The collection of these bits is the \emph{apartness table}; the abstract
structure it satisfies is a \emph{separoid}~\citep{arocha2002separoids,
strausz2004separoids}, and a crossing pair ($\sigma_{A,B}=0$) is exactly a
\emph{Radon partition}~\citep{radon1921}. This table is the relational payload
of the scene. Our thesis is that the table is a \emph{source}, a convex scene
is an \emph{encoder} of it, and a renderer is a \emph{channel} that carries it
to a decoder:
\begin{equation}
  \Sigma \;\xrightarrow{\ f_\theta\ }\; V \;\xrightarrow{\ \mathcal{R}_\omega\ }\;
  Y \;\xrightarrow{\ q_\phi\ }\; \hat\Sigma ,
  \label{eq:pipeline}
\end{equation}
where $V$ are scene vertices, $Y$ a rendered image, and $\hat\Sigma$ the decoded
apartness. \Cref{fig:pipeline} shows the pipeline on one scene.

\begin{figure}[t]
  \centering
  \includegraphics[width=\linewidth]{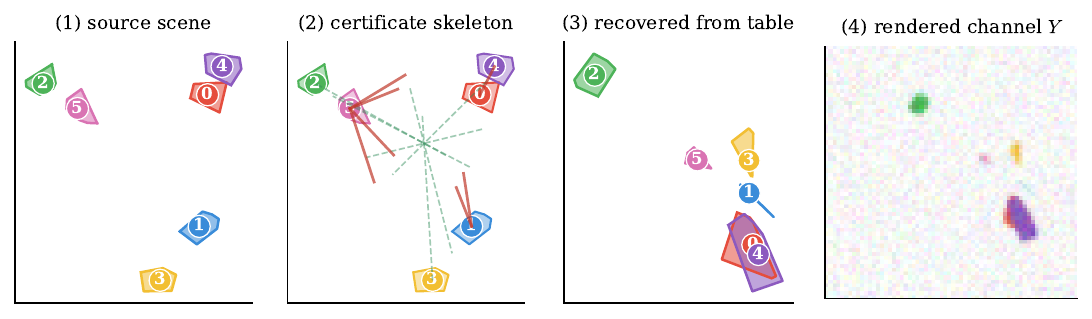}
  \caption{The separation-information rendering pipeline of
  \cref{eq:pipeline}. (1) A source convex scene induces an apartness table.
  (2) Its \emph{certificate skeleton}: minimal Radon partitions (solid red,
  crossings) and maximal separations (dashed green). (3) A scene optimised from
  the apartness \emph{table alone} reproduces the relational structure
  (\teaserAcc\% of apartness bits; apartness-equivalent, not shape-identical).
  (4) The rendered image $Y$ is the visual channel a decoder must read.}
  \label{fig:pipeline}
\end{figure}

\paragraph{From fitting to coding.} Merely placing bodies so that the apartness
holds is constraint satisfaction. The information-theoretic content appears once
we (a) charge a \emph{rate} for the geometry, (b) measure \emph{distortion} in
the currency of logical consequence rather than per-bit error, and (c) ask how
much of the apartness a rendered image actually \emph{conveys}. The three
ingredients (separation certificates, renderable convex carriers, and a
rate--distortion / mutual-information account of the channel) are individually
familiar, but to our knowledge they have not been combined: we are not aware of
a certificate-aware rate--distortion formulation for rendering separoid
apartness structures. Separoid theory rests on convex-set
representation~\citep{arocha2002separoids}; information theory is widely used in
visualization, e.g.\ viewpoint selection~\citep{vazquez2001viewpoint,
feixas2009unified}; and convex representation and decomposition are staple tools
in graphics~\citep{lien2007approximate,mamou2009vhacd}. We connect them.

\paragraph{Contributions.}
\begin{enumerate}[leftmargin=1.4em,itemsep=2pt]
\item \textbf{Apartness as a source, rendering as a channel}
(\Cref{sec:source,sec:channel}). We cast the apartness table as a binary source
whose sufficient statistic is the certificate skeleton (maximal separations and
minimal Radon partitions), and prove that this skeleton determines the entire
table by upward/downward closure (\Cref{prop:skeleton}).
\item \textbf{Certificate-aware rate--distortion} (\Cref{sec:distortion,sec:rate,sec:objective}).
We define a closure-aware distortion $\Dsep$ that weights each certificate by
the number of consequences it controls, a differentiable geometric rate $R(V)$,
and a single objective unifying separation, rate, directional information, and a
graphics-regularity prior. Apartness is realised through a differentiable
support-function channel, so the objective is end-to-end optimisable.
\item \textbf{Apartness mutual information and view selection}
(\Cref{sec:channel,sec:viewpoint}). We give a variational lower bound on
$I(\Sigma;Y)$ that a rendered scene attains, an associated capacity question for
a rendering family, and an information-optimal multi-view selection criterion;
the breadth of the separating-direction distribution is shown to be a
controllable knob for viewpoint robustness.
\item \textbf{Experiments} (\Cref{sec:experiments}). On planar convex scenes we
demonstrate near-perfect recovery from the table alone, a clean operational
rate--distortion frontier on which the consequence-weighted distortion is the
binding criterion, a measured information score for the rendered channel as a
function of resolution and noise, and the predicted view-selection and
robustness effects.
\end{enumerate}

\section{Related work}
\paragraph{Separoids, Radon partitions, realizability.}
Separoids axiomatise the apartness relation of convex
sets~\citep{arocha2002separoids,strausz2004separoids}; the dual notion, pairs
whose hulls intersect, is the classical Radon
partition~\citep{radon1921,tverberg1966}. Realizability of a separoid by a point
or convex-body configuration connects to oriented
matroids and multidimensional sorting~\citep{bjorner1999oriented,
goodman1983multidimensional}. We take this combinatorial structure as given and
ask a new, operational question: how cheaply can a \emph{rendered} convex scene
transmit it.

\paragraph{Differentiable rendering and convex geometry in graphics.}
Differentiable renderers make image formation amenable to gradient-based
optimisation~\citep{loper2014opendr,kato2018neural,liu2019softras,laine2020modular,
kato2020differentiable}. Convex bodies and (approximate) convex decompositions
are standard primitives~\citep{lien2007approximate,mamou2009vhacd}, and the
support function is the canonical differentiable handle on a convex
shape~\citep{schneider2014convex}. Our renderer is a soft support-function
rasteriser; what is new is the \emph{objective}, which optimises a scene to be a
low-rate, faithful carrier of a discrete relational structure.

\paragraph{Information theory in visualization and learning.}
Viewpoint entropy and related information measures drive view selection and
saliency~\citep{vazquez2001viewpoint,feixas2009unified}. Rate--distortion
theory~\citep{shannon1959rate,berger1971rate,cover2006elements} and variational
mutual-information bounds~\citep{barber2003im,poole2019variational,
alemi2017deep,tishby1999information} underpin our channel formulation. We use
these tools not for pixels or features but for a logical certificate carried by
geometry.

\section{Apartness structures as a source}
\label{sec:source}
Let $E=\{1,\dots,n\}$ index the sites (bodies) and let
\[
  \dom=\{(A,B): A,B\subseteq E,\ A,B\neq\varnothing,\ A\cap B=\varnothing\}
\]
be the disjoint pairs. The \emph{apartness table} of a scene is
$\Sigma=\{\sigma_{A,B}\}_{(A,B)\in\dom}$, with $\sigma_{A,B}=1$ iff the group
hulls are disjoint. As an abstract object $\Sigma$ is a
separoid~\citep{arocha2002separoids}: the relation $\sep$ defined by
$A\sep B \Leftrightarrow \sigma_{A,B}=1$ is symmetric, every separated pair is
disjoint, and it is closed under taking subsets,
\begin{equation}
  A\sep B,\ A'\subseteq A,\ B'\subseteq B \ \Longrightarrow\ A'\sep B'.
  \label{eq:downclosed}
\end{equation}
We regard $\Sigma$ as a binary \textbf{source} to be transmitted.

\paragraph{Certificate skeleton.}
Two families of pairs summarise $\Sigma$. A \emph{maximal separation} is a
separated pair $(A,B)$ with $\sigma_{A,B}=1$ that admits no single-site
extension staying separated; write $\Mset(\Sigma)$ for their set. A
\emph{minimal Radon partition} is a crossing pair $(A,B)$ with $\sigma_{A,B}=0$
such that removing any one site (keeping both parts non-empty) yields a
separation; write $\Pset(\Sigma)$. Because separations are downward closed
\eqref{eq:downclosed} and crossings are upward closed (enlarging either part
keeps the hulls intersecting), these two antichains generate the table.

\begin{proposition}[The skeleton determines the table]
\label{prop:skeleton}
For every disjoint pair $(A,B)\in\dom$,
\[
  \sigma_{A,B}=1 \iff (A,B)\sqsubseteq (M_1,M_2)\ \text{for some}\ (M_1,M_2)\in\Mset(\Sigma),
\]
where $(A,B)\sqsubseteq(M_1,M_2)$ means $A\subseteq M_1,\ B\subseteq M_2$.
Equivalently, $\sigma_{A,B}=0$ iff $(A,B)$ contains some
$(P_1,P_2)\in\Pset(\Sigma)$. Hence $\Sigma$ is recoverable from
$(\Mset,\Pset)$ alone.
\end{proposition}

\begin{proof}
If $(A,B)\sqsubseteq(M_1,M_2)$ with $\sigma_{M_1,M_2}=1$, then
\eqref{eq:downclosed} gives $\sigma_{A,B}=1$. Conversely, if $\sigma_{A,B}=1$,
repeatedly add a free site to $A$ or to $B$ whenever the pair stays separated;
the process terminates (only finitely many sites) at a maximal separation
containing $(A,B)$. The crossing statement is the contrapositive together with
upward closure of crossings, by the same single-site argument applied to
removals.
\end{proof}

Thus the source has a sufficient statistic
$Z_\Sigma=(M_\Sigma,P_\Sigma)$ formed by the skeleton bits, and we may transmit
$Z_\Sigma$ in place of $\Sigma$. In our scenes the skeleton is far smaller than
the table (\Cref{tab:sizes}), which is what makes the certificate view
worthwhile.

\paragraph{Source entropy.} This redundancy can be quantified. The apartness
table $\Sigma$ has $|\dom|=\tfrac12(3^n-2^{n+1}+1)$ entries, but its entropy is
far lower: when the bodies are points, $\Sigma$ is a function of the
configuration's order type (oriented matroid), since the Radon partitions are
fixed by the chirotope~\citep{bjorner1999oriented}. As there are only
$2^{\Theta(n\log n)}$ order types of $n$ points in the plane~\citep{goodman1983multidimensional},
the source obeys
\begin{equation}
  H(\Sigma)\;=\;O(n\log n)\quad\text{bits},
  \label{eq:srcentropy}
\end{equation}
exponentially below the table size. The certificate skeleton is the structural
witness of this gap, and the rate--distortion and channel results below turn it
into operational statements: a faithful carrier and a good view need only convey
these $O(n\log n)$ bits, not the full table.

\section{Rendering as a channel}
\label{sec:channel}
\paragraph{Encoder.} A scene is $V=\{v_{e,j}\}$, $e\in E$, $j\in\{1,\dots,m\}$,
with $C_e(V)=\conv\{v_{e,1},\dots,v_{e,m}\}\subset\R^d$. The encoder
$f_\theta:\Sigma\mapsto V$ produces a scene realising the source; in this work
$f_\theta$ is the optimisation of \Cref{sec:objective}.

\paragraph{Renderer / channel.} A renderer $\mathcal{R}_\omega$ with parameters
$\omega$ (camera, resolution, palette, sensor) maps a scene to an image,
$Y\sim p_\omega(\,\cdot\mid V)$. We use a soft support-function rasteriser
(\Cref{sec:objective}): a point $x$ lies in $C_e$ iff
$\langle u,x\rangle\le h_e(u)$ for all directions $u$, where
$h_e(u)=\max_j\langle u,v_{e,j}\rangle$ is the support
function~\citep{schneider2014convex}; sites are alpha-composited under a fixed
palette and corrupted by Gaussian sensor noise. Occlusion, finite resolution,
colour collisions, and noise make $Y$ a genuinely lossy channel.

\paragraph{Decoder and apartness mutual information.}
A decoder $q_\phi(\hat\Sigma\mid Y)=\prod_{(A,B)}q_\phi(\hat\sigma_{A,B}\mid Y)$
estimates apartness bits from the image. We want $Y$ to carry as much of the
source as possible, i.e.\ to maximise $I(\Sigma;Y)$ (or, for any chosen
sub-structure such as the certificate $Z_\Sigma$ or the pairwise apartness graph
of \Cref{sec:experiments}, the corresponding $I(\cdot\,;Y)$). Using the standard
variational bound~\citep{barber2003im,poole2019variational},
\begin{equation}
  I(Z_\Sigma;Y)\;=\;H(Z_\Sigma)-H(Z_\Sigma\mid Y)\;\ge\;
  H(Z_\Sigma)+\E_{Z,Y}\big[\log q_\phi(Z_\Sigma\mid Y)\big],
  \label{eq:mibound}
\end{equation}
which is tight when $q_\phi$ matches the true posterior. The bound turns
maximising apartness information into minimising the decoder cross-entropy, with
the operational reading: \emph{how many bits of apartness does the rendered
image deliver?} We report the normalised
\begin{equation}
  \mathrm{IS}\;=\;\frac{\hat I(\Sigma;Y)}{H(\Sigma)}\in[0,1],
  \qquad
  \hat I(\Sigma;Y)=H(\Sigma)-\hat H(\Sigma\mid Y),
  \label{eq:infoscore}
\end{equation}
the fraction of source entropy the channel transmits, with $H(\Sigma)$ and
$\hat H(\Sigma\mid Y)$ estimated as sums of per-bit Bernoulli entropies.

\paragraph{What the pipeline can and cannot do.} The encoding, rendering and
decoding of \cref{eq:pipeline} form a Markov chain
$\Sigma\to V\to Y\to\hat\Sigma$, so the data-processing
inequality~\citep{cover2006elements} bounds every stage at once,
\begin{equation}
  I(\Sigma;\hat\Sigma)\;\le\;I(\Sigma;Y)\;\le\;I(\Sigma;V)\;\le\;H(\Sigma).
  \label{eq:dpi}
\end{equation}
No decoder recovers more apartness than the rendered image carries, and no
renderer carries more than the scene encodes. Two design goals follow: the
encoder $f_\theta$ should keep $I(\Sigma;V)$ near $H(\Sigma)$ (lose no structure
in the geometry), and the renderer should keep $I(\Sigma;Y)$ near $I(\Sigma;V)$
(lose no structure in the picture). The information score IS estimates the
middle term normalised by $H(\Sigma)$ and is therefore a property of the
encoder--channel pair, not of any particular decoder.

\section{Certificate-aware distortion}
\label{sec:distortion}
A naive distortion is the Hamming error over the table,
$d_H(\Sigma,\hat\Sigma)=\sum_{(A,B)\in\dom}\indic{\sigma_{A,B}\neq\hat\sigma_{A,B}}$.
This is misleading: by \Cref{prop:skeleton} a single wrong \emph{maximal}
separation falsifies all separations beneath it, and a wrong \emph{minimal}
Radon partition falsifies all crossings above it. Distortion should be measured
in logical consequences. Define the consequence weights
\begin{equation}
  w_M=\big|\Down(A,B)\big|=(2^{|A|}-1)(2^{|B|}-1),
  \qquad
  w_P=\big|\Up(A,B)\big|=3^{\,n-|A|-|B|},
  \label{eq:weights}
\end{equation}
counting, respectively, the non-empty sub-pairs a separation implies and the
extensions a crossing forces. The \emph{certificate (closure-aware) distortion}
is
\begin{equation}
  \Dsep(\Sigma,\hat\Sigma)=
  \sum_{(A,B)\in\Mset(\Sigma)} w_M\,\indic{\hat\sigma_{A,B}\neq 1}
  +\sum_{(A,B)\in\Pset(\Sigma)} w_P\,\indic{\hat\sigma_{A,B}\neq 0}.
  \label{eq:dsep}
\end{equation}
Its differentiable surrogate replaces each indicator by the soft separation
hinge of \Cref{sec:objective}, so that distortion is \emph{logical-consequence
loss}, not pixel error.

\section{Geometry rate}
\label{sec:rate}
The rate is the description length of the carrier $V$. Counting vertices is too
crude; we use a differentiable geometric code length with three terms:
\begin{equation}
  R(V)=\underbrace{\sum_{e,j}\log\!\Big(1+\tfrac{\|v_{e,j}\|^2}{\Delta_q^2}\Big)}_{\text{quantised coordinates}}
  \;+\;\underbrace{\sum_e H(\kappa_e)}_{\text{boundary-curvature entropy}},
  \label{eq:rate}
\end{equation}
where $\Delta_q$ is the coordinate quantisation step. The first term is the cost
of locating quantised vertices; the second is the entropy of the discrete
turning-angle (curvature) profile of each body,
$\kappa_{e,j}=\angle(v_{e,j+1}-v_{e,j})-\angle(v_{e,j}-v_{e,j-1})$, normalised to
$p_{e,j}=|\kappa_{e,j}|/\sum_i|\kappa_{e,i}|$ and
$H(\kappa_e)=-\sum_j p_{e,j}\log p_{e,j}$. Low rate favours compact, smooth
silhouettes that still encode the apartness, exactly the graphics intuition that
a clean figure should not waste boundary complexity.

\section{The objective and its differentiable channel}
\label{sec:objective}
\paragraph{Directional separation channel.}
For unit direction $u$, the group support is $h_A(u)=\max_{a\in A}h_a(u)$ and the
signed margin
\begin{equation}
  s_{A,B}(u)=-\,h_A(u)-h_B(-u)
  \label{eq:margin}
\end{equation}
is positive iff $u$ separates the groups with that margin (a separating
hyperplane normal). Sampling $L$ directions $\{u_\ell\}$, a smooth max margin and
a direction distribution are
\begin{equation}
  \Delta_{A,B}=\tau\log\!\sum_{\ell=1}^{L}\exp\!\Big(\tfrac{s_{A,B}(u_\ell)}{\tau}\Big),
  \qquad
  p_{A,B}(\ell)=\softmax_\ell\!\Big(\tfrac{s_{A,B}(u_\ell)}{\tau}\Big).
  \label{eq:softsep}
\end{equation}
As $\tau\to 0$, $\Delta_{A,B}\to\max_u s_{A,B}(u)$, which is positive exactly when
the groups are separable. The \emph{apartness of a pair is thus encoded as a
distribution $p_{A,B}$ over separating directions}, with entropy
\begin{equation}
  H_{\mathrm{dir}}(A,B)=-\sum_{\ell=1}^{L}p_{A,B}(\ell)\log p_{A,B}(\ell)
  \label{eq:hdir}
\end{equation}
measuring how broadly the separation is witnessed. The set
$\{u:s_{A,B}(u)>0\}$ is the \emph{witnessing cone} of the separation; its
breadth is what governs recoverability across viewpoints (\Cref{sec:viewpoint}),
and we widen it directly with the robustness term below.

\paragraph{Separation and certificate losses.} With a target margin $\gamma$,
\begin{equation}
  \mathcal{L}_{\mathrm{sep}}=
  \!\!\sum_{(A,B):\sigma=1}\!\!\big[\max(0,\gamma-\Delta_{A,B})\big]^2
  +\!\!\sum_{(A,B):\sigma=0}\!\!\big[\max(0,\Delta_{A,B})\big]^2,
  \label{eq:lsep}
\end{equation}
and the certificate-weighted version applies the weights \eqref{eq:weights} over
$\Mset\cup\Pset$, the differentiable surrogate of $\Dsep$.

\paragraph{Radon witnesses.} A crossing $(A,B)\in\Pset$ can be certified
constructively by a witness point in both hulls: barycentric weights
$\alpha,\beta\ge0$ with $\sum\alpha=\sum\beta=1$ and
$\sum_{a,i}\alpha^P_{a,i}v_{a,i}=\sum_{b,j}\beta^P_{b,j}v_{b,j}$. Penalising the
squared mismatch yields $\mathcal{L}_{\mathrm{radon}}$, and the witness entropy
$H(\alpha),H(\beta)$ selects between a \emph{sparse} certificate (few vertices)
and \emph{distributed visual evidence} (a crossing spread over a legible region
rather than a single pixel of contact).

\paragraph{Viewpoint robustness.} To make a separation readable from many
directions we widen its witnessing cone, rewarding the soft fraction of
directions that already separate the pair,
\begin{equation}
  R_{\mathrm{view}}(V)=\!\!\sum_{(A,B):\sigma=1}\!\!
  \Big(1-\tfrac{1}{L}\textstyle\sum_{\ell}\varsigma\!\big(s_{A,B}(u_\ell)/\tau_C\big)\Big),
  \label{eq:rview}
\end{equation}
with $\varsigma$ the logistic. Minimising $R_{\mathrm{view}}$ spreads the
apartness over more separating directions (the broad, high-entropy regime of
the directional distribution $p_{A,B}$) at a small geometric-rate cost.

\paragraph{Graphics regularity.} A shape prior $G(V)$ favours compact,
near-regular bodies by penalising the squared coefficient of variation of vertex
radii about each centroid, keeping silhouettes legible.

\paragraph{Full objective.} Putting the pieces together, with
$V=f_\theta(\Sigma)$ and $Y\sim p_\omega(\cdot\mid V)$,
\begin{equation}
  \boxed{\;
  \mathcal{J}=
  \underbrace{w_{\mathrm{sep}}\mathcal{L}_{\mathrm{sep}}
  +w_{\mathrm{cert}}\mathcal{L}_{\mathrm{cert}}}_{\text{distortion}}
  \;+\;\lambda_R\,R(V)\;+\;\lambda_G\,G(V)\;+\;\lambda_C\,R_{\mathrm{view}}(V) ,
  \;}
  \label{eq:objective}
\end{equation}
optionally augmented by $-\lambda_I\,\E_{Z,Y}\log q_\phi(Z\mid Y)$ to push the
rendered channel toward maximal apartness information via \eqref{eq:mibound}. The
Lagrangian form $\mathcal{J}_\beta=\E[\Dsep]+\beta\,\E[R(V)]$ traces the
operational rate--distortion curve as $\beta$ (here $\lambda_R$) varies, while
$\lambda_C$ trades a little rate for viewpoint robustness.

\section{Information-optimal viewpoints}
\label{sec:viewpoint}
A single orthographic view along direction $u$ observes each body's
\emph{shadow}, the support interval $[\,{-h_e(-u)},\,h_e(u)\,]$. A separation
$(A,B)$ is \emph{certified} by view $u$ iff the group shadows are disjoint,
i.e.\ iff $s_{A,B}(u)>0$; a view never reports a false separation, so it recovers
exactly the separations whose witnessing cone contains $u$. The
\emph{viewpoint informativeness} is the fraction of separations a view certifies,
\begin{equation}
  \mathcal{I}(\omega)=I(Z_\Sigma;Y_\omega),\qquad
  \widehat{\mathcal{I}}(\omega)=\E_{Z,Y_\omega}\log q_\phi(Z_\Sigma\mid Y_\omega),
  \label{eq:viewinfo}
\end{equation}
and the best view is $\omega^\star=\argmax_\omega\widehat{\mathcal{I}}(\omega)$.
For several views, the information-optimal set trades coverage against
redundancy,
\begin{equation}
  \Omega_k^\star=\argmax_{|\Omega_k|=k} I\big(Z_\Sigma;Y_{\omega_1},\dots,Y_{\omega_k}\big)
  -\eta\sum_{i<j}I\big(Y_{\omega_i};Y_{\omega_j}\big),
  \label{eq:multiview}
\end{equation}
which we approximate greedily by maximal coverage of as-yet-uncertified
separations. Finally, for a rendering family
$\mathcal{F}=\{f_\theta,\mathcal{R}_\omega\}$ under a rate budget $R_0$, the
\emph{apartness channel capacity}
\begin{equation}
  C_{\mathcal{F}}(R_0)=\sup_{\substack{p(\Sigma),\,f,\,\mathcal{R}\\ \E R(V)\le R_0}}
  I(Z_\Sigma;Y)
  \label{eq:capacity}
\end{equation}
asks how many bits of separation information a renderable convex scene family can
carry, the natural endpoint of the channel view.

\section{Experiments}
\label{sec:experiments}
\paragraph{Setup.} We study planar scenes ($d=2$) with $n\in\{5,6,7\}$ convex
polygons of $m=8$ vertices. Ground-truth tables are computed exactly by
linear-programming hull-intersection tests; the certificate skeleton is
extracted by single-site extension/removal tests (\Cref{prop:skeleton}). The
reconstruction encoder optimises \eqref{eq:objective} with Adam~\citep{kingma2015adam}
over $L{=}96$ directions, annealing $\tau$ from $0.25$ to $0.05$, with multiple
restarts kept by the realisability loss only (never peeking at the held-out
score). The renderer outputs RGB images with Gaussian noise. For the channel
experiment the decoding target is the \emph{pairwise apartness graph} (the
$\binom{n}{2}$ singleton--singleton bits $\sigma_{\{i\},\{j\}}$, the directly
visible $1$-skeleton of the structure), and the decoder is object-centric: it
colour-keys the image into one soft response map per body, reduces each to
spatial moments (visible mass, centroid, spread), and a shared pairwise MLP
reads apartness from these summaries, so resolution and noise act precisely by
degrading the moment estimates. It is trained with binary cross-entropy, i.e.\
maximising the bound \eqref{eq:mibound}. Everything runs on a single laptop.

\begin{table}[t]
  \centering
  \caption{Source vs.\ certificate size, and reconstruction accuracy
  (mean over scenes). The skeleton is far smaller than the table yet,
  optimised alone, recovers it almost perfectly.}
  \label{tab:sizes}
  \small
  \begin{tabular}{lcccccc}
    \toprule
    & \multicolumn{2}{c}{size} & \multicolumn{2}{c}{full-table loss}
      & \multicolumn{2}{c}{skeleton-only loss}\\
    \cmidrule(lr){2-3}\cmidrule(lr){4-5}\cmidrule(lr){6-7}
    $n$ & $|\dom|$ & $|\Mset|{+}|\Pset|$
        & table acc. & cert.\ acc.
        & table acc. & cert.\ acc.\\
    \midrule
    5 & \tabDEfive  & \tabSKfive  & \tabFTtableFive  & \tabFTcertFive  & \tabSKtableFive  & \tabSKcertFive\\
    6 & \tabDEsix   & \tabSKsix   & \tabFTtableSix   & \tabFTcertSix   & \tabSKtableSix   & \tabSKcertSix\\
    7 & \tabDEseven & \tabSKseven & \tabFTtableSeven & \tabFTcertSeven & \tabSKtableSeven & \tabSKcertSeven\\
    \bottomrule
  \end{tabular}
\end{table}

\paragraph{(A) Recovery from the table, and the sufficiency of the skeleton.}
\Cref{fig:reconstruction} and \Cref{tab:sizes} confirm \Cref{prop:skeleton}
operationally: optimising \emph{only} the certificate skeleton recovers the full
apartness table to within a fraction of a percent of optimising the entire
table, despite the skeleton being many times smaller. Certificate weighting
concentrates whatever residual error remains away from high-consequence bits, so
certificate accuracy stays high even when a few low-consequence bits slip.

\begin{figure}[t]
  \centering
  \includegraphics[width=\linewidth]{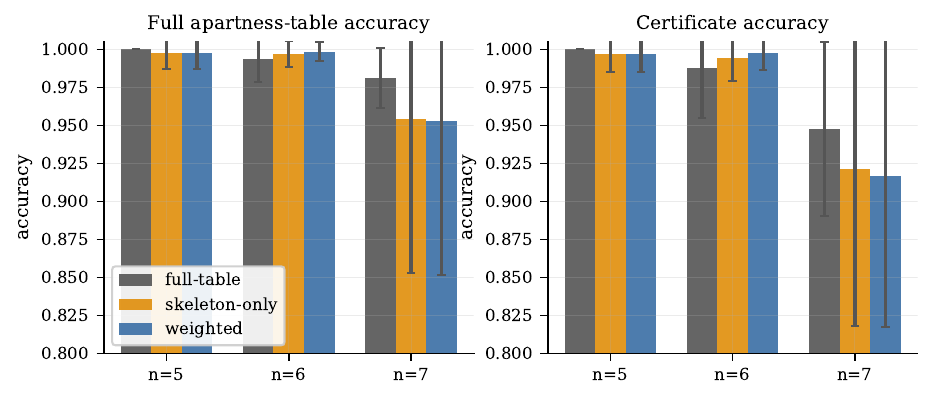}
  \caption{Reconstruction accuracy by scene size for three objective variants.
  Optimising the small certificate skeleton (orange) matches optimising the full
  table (grey) on both the full-table metric (left) and the certificate metric
  (right); the weighted objective (blue) protects high-consequence bits.}
  \label{fig:reconstruction}
\end{figure}

\paragraph{(B) Operational rate--distortion.}
We reconstruct a faithful scene from each table, then quantise its vertices to a
grid of step $\Delta$: coarser grids cost fewer bits but flip more apartness
bits. \Cref{fig:rd}a is the resulting frontier: a clean, monotone trade-off of
geometry rate (in bits) against certificate distortion $\Dsep$, with a wide
high-rate plateau (quantisation is harmless until $\Delta$ approaches the
separation margins) followed by a steep collapse. \Cref{fig:rd}b contrasts the
two fidelity criteria on the same scenes: the consequence-weighted certificate
distortion sits well above plain Hamming error and rises faster as the rate
drops, because coarse quantisation preferentially flips the high-consequence
certificate bits. This is the central quantitative claim: \emph{apartness-preserving
rendering behaves like a source code, and the fidelity criterion that matters is
the structured one.}

\begin{figure}[t]
  \centering
  \includegraphics[width=\linewidth]{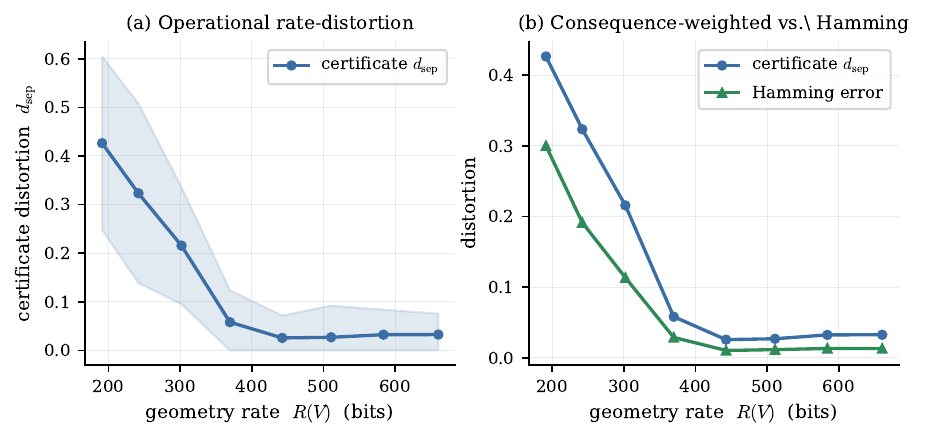}
  \caption{(a) Operational rate--distortion frontier from coordinate
  quantisation: geometry rate (bits) vs.\ certificate distortion $\Dsep$ (shaded
  band: $\pm1$ std over scenes). (b) The consequence-weighted certificate
  distortion exceeds plain Hamming error and degrades faster at low rate.}
  \label{fig:rd}
\end{figure}

\paragraph{(C) The rendered channel and its information score.}
Training the object-centric decoder on rendered images and evaluating the bound
\eqref{eq:infoscore} measures how much of the pairwise apartness graph the
picture conveys. \Cref{fig:channel} shows the information score rising with image
resolution (more visual rate) and collapsing with sensor noise, exactly the
behaviour of a capacity-limited channel \eqref{eq:capacity}. At $48\times48$ with
mild noise the channel delivers $\mathrm{IS}=\channelISfortyeight$ of the
apartness-graph entropy at \channelAccfortyeight{}\% per-bit accuracy; sample
renders appear in \Cref{fig:samples}. This descent is dictated by the channel,
not by our network. Fano's inequality~\citep{fano1961transmission,cover2006elements}
forces any reader of the image to a per-bit error of at least
$p_e\ge H_b^{-1}\big((1-\mathrm{IS})\,H(\sigma)\big)$ on a balanced apartness bit,
where $H_b^{-1}$ is the inverse binary entropy on $[0,\tfrac12]$: the falling
noise curve of \Cref{fig:channel}b is therefore a hard error floor that no
decoder can beat, and the accuracy we observe sits just above it, so the decoder
operates near the channel limit.

\begin{figure}[t]
  \centering
  \includegraphics[width=\linewidth]{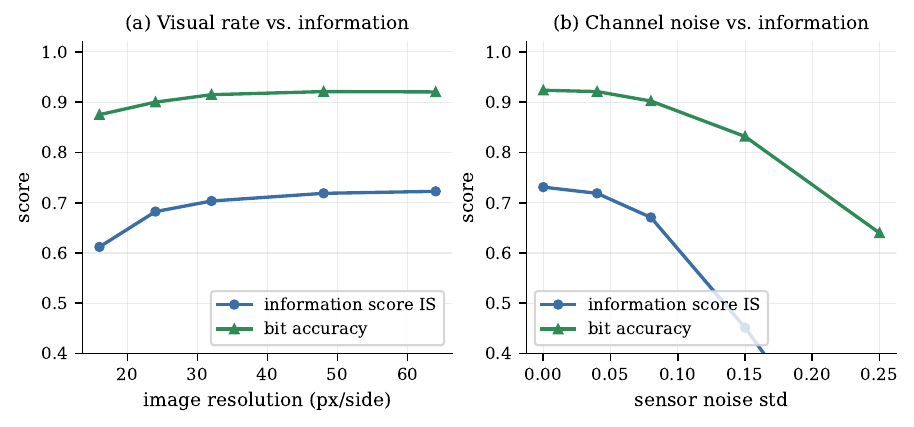}
  \caption{The rendered visual channel. (a) Information score and bit accuracy
  rise with image resolution (visual rate) and saturate. (b) Both fall as sensor
  noise grows (the information score collapses toward zero), the signature of a
  capacity-limited channel.}
  \label{fig:channel}
\end{figure}

\begin{figure}[t]
  \centering
  \includegraphics[width=\linewidth]{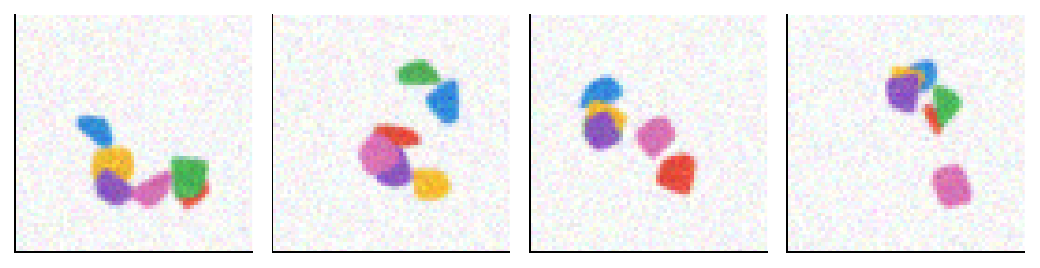}
  \caption{Rendered convex apartness scenes: the channel output $Y$ a decoder
  must read. Apartness (clusters vs.\ isolated bodies) is visible but corrupted
  by occlusion, finite resolution, colour collisions, and noise.}
  \label{fig:samples}
\end{figure}

\paragraph{(D) Information-optimal viewpoints.}
\Cref{fig:viewpoint}(a) plots viewpoint informativeness $I(\theta)$: certain
directions certify many more separations than others. A single best view already
certifies \viewBestSingle{} of all separations on average; greedy
information-optimal view selection \eqref{eq:multiview} reaches full coverage in
far fewer views than random selection (\Cref{fig:viewpoint}b), and the
histogram of witnessing-cone widths (\Cref{fig:viewpoint}c) explains why: most
separations are robust across a broad cone of viewpoints, while a minority are
fragile. Greedy is not merely expedient here. The map from a view set $\Omega$ to
the separations it certifies,
$\mathrm{cov}(\Omega)=\{(A,B):\exists\,u\in\Omega,\ s_{A,B}(u)>0\}$, is a coverage
function, hence monotone and submodular, so the greedy selector of
\eqref{eq:multiview} certifies at least a $(1-1/e)$ fraction of the optimum
reachable by \emph{any} $k$ views~\citep{nemhauser1978analysis}, precisely the
gap over random selection that \Cref{fig:viewpoint}b traces. Choosing views for
apartness is thus an instance of submodular maximisation, the same structure
that governs sensor placement and view planning~\citep{krause2008near}.

\begin{figure}[t]
  \centering
  \includegraphics[width=\linewidth]{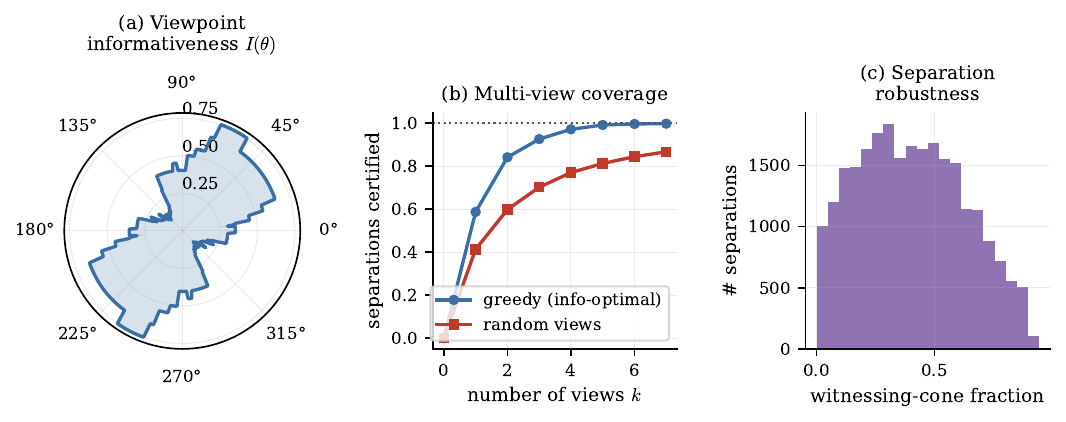}
  \caption{Viewpoint information. (a) Polar informativeness $I(\theta)$.
  (b) Greedy info-optimal views certify all separations far faster than random
  views. (c) Distribution of witnessing-cone widths (per-separation robustness).}
  \label{fig:viewpoint}
\end{figure}

\paragraph{(E) Trading rate for viewpoint robustness.}
The encoder can deliberately make separations easier to see. Increasing the
robustness weight $\lambda_C$ in \eqref{eq:objective} widens each separation's
witnessing cone monotonically (\Cref{fig:entropy}a), spreading the apartness over
more separating directions. \Cref{fig:entropy}b shows the price: a modest
increase in geometry rate $R(V)$, while table accuracy is essentially preserved.
This is a clean rate--robustness knob: a concrete bridge from the
directional-information term to a graphics outcome (figures readable from more
viewpoints).

\begin{figure}[t]
  \centering
  \includegraphics[width=\linewidth]{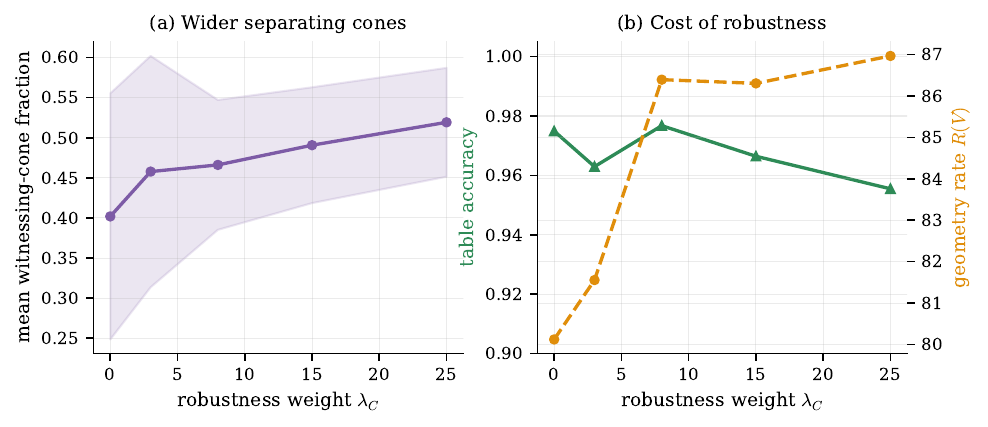}
  \caption{Trading rate for viewpoint robustness. (a) Raising the robustness
  weight $\lambda_C$ widens separating cones (mean cone fraction and best
  single-view coverage both rise). (b) The cost is a modest rate increase, with
  table accuracy essentially unchanged.}
  \label{fig:entropy}
\end{figure}

\paragraph{Apartness as arcs on the view circle.} The effect in
\Cref{fig:entropy} has a clean geometric reading. Each separation's witnessing
cone $\{u:s_{A,B}(u)>0\}$ is an arc of the circle of viewing directions, so
certifying every separation is a covering problem on $S^1$: the chosen views must
stab all arcs. If every arc has angular measure at least $\gamma\cdot 2\pi$, then
$\lceil 1/\gamma\rceil$ equispaced views already form such a cover, and the
robustness term \eqref{eq:rview} acts by enlarging the arcs (raising $\gamma$),
which lowers this bound; sharper guarantees follow from $\varepsilon$-net theory
for geometric range spaces~\citep{haussler1987epsilon}. Reading apartness as a
measure on arcs of $S^1$ places view planning for separation structures within
integral geometry~\citep{santalo2004integral}, a connection we expect to repay
further study.

\section{Discussion and limitations}
Our scenes are planar and modest in size; computing exact tables and skeletons
is exponential in $n$, so the source itself is the scaling bottleneck, not the
optimisation. Reconstruction is a non-convex problem and occasionally lands in a
local optimum at $n=7$, which restarts mitigate but do not eliminate. The
renderer is deliberately simple; richer cameras, true 3D bodies, and learned
renderers fit the same channel formulation and are natural extensions. The
capacity \eqref{eq:capacity} is defined but only bounded empirically here.
Finally, our distortion privileges the certificate skeleton; other consequence
weightings (e.g.\ task-specific) slot into the same framework.

\section{Conclusion}
We recast rendering a convex scene as transmitting its apartness structure over
a visual channel, and posed apartness-preserving rendering as a rate--distortion
problem with a closure-aware fidelity criterion. Maximal separations and minimal
Radon partitions are the certificates; a differentiable support-function scene is
the low-rate carrier; and the rendered image is the channel whose apartness
mutual information we lower-bound and measure. Experiments confirm recovery from
the table alone, a clean rate--distortion frontier on which the
consequence-weighted distortion is the binding criterion, a measurable
information score for rendered scenes, and an information-theoretic account of
viewpoint selection and robustness. We hope the
certificate-as-payload view encourages rendering objectives that are honest about
\emph{what a picture is supposed to make recoverable.}

\paragraph{Reproducibility.} The complete implementation, experiment drivers,
and result files accompany this paper as ancillary material; running
\texttt{anc/experiments/run\_all.py} regenerates every figure and table from
scratch with fixed seeds.

\small
\bibliographystyle{plainnat}
\bibliography{references}

\end{document}

%% file: results_macros.tex
\newcommand{\tabDEfive}{90}
\newcommand{\tabSKfive}{9.2}
\newcommand{\tabFTtableFive}{100.0}
\newcommand{\tabFTcertFive}{100.0}
\newcommand{\tabSKtableFive}{99.8}
\newcommand{\tabSKcertFive}{99.7}
\newcommand{\tabDEsix}{301}
\newcommand{\tabSKsix}{14.8}
\newcommand{\tabFTtableSix}{99.4}
\newcommand{\tabFTcertSix}{98.8}
\newcommand{\tabSKtableSix}{99.7}
\newcommand{\tabSKcertSix}{99.5}
\newcommand{\tabDEseven}{966}
\newcommand{\tabSKseven}{28.0}
\newcommand{\tabFTtableSeven}{98.1}
\newcommand{\tabFTcertSeven}{94.8}
\newcommand{\tabSKtableSeven}{95.4}
\newcommand{\tabSKcertSeven}{92.1}
\newcommand{\reconHamSix}{99.9}
\newcommand{\channelISfortyeight}{0.72}
\newcommand{\channelAccfortyeight}{92.1}
\newcommand{\viewBestSingle}{58.7\%}
\newcommand{\teaserAcc}{97.3}

%% file: references.bib
@article{radon1921,
  author  = {Radon, Johann},
  title   = {Mengen konvexer K\"orper, die einen gemeinsamen Punkt enthalten},
  journal = {Mathematische Annalen},
  volume  = {83},
  number  = {1--2},
  pages   = {113--115},
  year    = {1921},
  doi     = {10.1007/BF01464231}
}

@article{arocha2002separoids,
  author  = {Arocha, Jorge Luis and Bracho, Javier and Montejano, Luis and
             Oliveros, Deborah and Strausz, Ricardo},
  title   = {Separoids, their categories and a {H}adwiger-type theorem for
             transversals},
  journal = {Discrete \& Computational Geometry},
  volume  = {27},
  number  = {3},
  pages   = {377--385},
  year    = {2002},
  doi     = {10.1007/s00454-001-0075-2}
}

@phdthesis{strausz2004separoids,
  author = {Strausz, Ricardo},
  title  = {Separoides},
  school = {Universidad Nacional Aut\'onoma de M\'exico},
  year   = {2004}
}

@book{bjorner1999oriented,
  author    = {Bj\"orner, Anders and Las Vergnas, Michel and Sturmfels, Bernd
               and White, Neil and Ziegler, G\"unter M.},
  title     = {Oriented Matroids},
  edition   = {2},
  publisher = {Cambridge University Press},
  series    = {Encyclopedia of Mathematics and its Applications},
  volume    = {46},
  year      = {1999}
}

@article{tverberg1966,
  author  = {Tverberg, Helge},
  title   = {A generalization of {R}adon's theorem},
  journal = {Journal of the London Mathematical Society},
  volume  = {s1-41},
  number  = {1},
  pages   = {123--128},
  year    = {1966}
}

@book{schneider2014convex,
  author    = {Schneider, Rolf},
  title     = {Convex Bodies: The {B}runn--{M}inkowski Theory},
  edition   = {2},
  publisher = {Cambridge University Press},
  year      = {2014}
}

@book{fano1961transmission,
  author    = {Fano, Robert M.},
  title     = {Transmission of Information: A Statistical Theory of
               Communications},
  publisher = {MIT Press},
  year      = {1961}
}

@book{cover2006elements,
  author    = {Cover, Thomas M. and Thomas, Joy A.},
  title     = {Elements of Information Theory},
  edition   = {2},
  publisher = {Wiley-Interscience},
  year      = {2006}
}

@incollection{shannon1959rate,
  author    = {Shannon, Claude E.},
  title     = {Coding theorems for a discrete source with a fidelity criterion},
  booktitle = {IRE National Convention Record, Part 4},
  pages     = {142--163},
  year      = {1959}
}

@book{berger1971rate,
  author    = {Berger, Toby},
  title     = {Rate Distortion Theory: A Mathematical Basis for Data Compression},
  publisher = {Prentice-Hall},
  year      = {1971}
}

@inproceedings{barber2003im,
  author    = {Barber, David and Agakov, Felix},
  title     = {The {IM} algorithm: a variational approach to information
               maximization},
  booktitle = {Advances in Neural Information Processing Systems (NeurIPS)},
  year      = {2003}
}

@inproceedings{poole2019variational,
  author    = {Poole, Ben and Ozair, Sherjil and van den Oord, Aaron and
               Alemi, Alexander A. and Tucker, George},
  title     = {On variational bounds of mutual information},
  booktitle = {International Conference on Machine Learning (ICML)},
  year      = {2019}
}

@inproceedings{alemi2017deep,
  author    = {Alemi, Alexander A. and Fischer, Ian and Dillon, Joshua V. and
               Murphy, Kevin},
  title     = {Deep variational information bottleneck},
  booktitle = {International Conference on Learning Representations (ICLR)},
  year      = {2017}
}

@inproceedings{tishby1999information,
  author    = {Tishby, Naftali and Pereira, Fernando C. and Bialek, William},
  title     = {The information bottleneck method},
  booktitle = {37th Annual Allerton Conference on Communication, Control, and
               Computing},
  pages     = {368--377},
  year      = {1999}
}

@inproceedings{loper2014opendr,
  author    = {Loper, Matthew M. and Black, Michael J.},
  title     = {{OpenDR}: An approximate differentiable renderer},
  booktitle = {European Conference on Computer Vision (ECCV)},
  pages     = {154--169},
  year      = {2014}
}

@inproceedings{kato2018neural,
  author    = {Kato, Hiroharu and Ushiku, Yoshitaka and Harada, Tatsuya},
  title     = {Neural 3{D} mesh renderer},
  booktitle = {IEEE Conference on Computer Vision and Pattern Recognition (CVPR)},
  pages     = {3907--3916},
  year      = {2018}
}

@inproceedings{liu2019softras,
  author    = {Liu, Shichen and Li, Tianye and Chen, Weikai and Li, Hao},
  title     = {Soft Rasterizer: A differentiable renderer for image-based 3{D}
               reasoning},
  booktitle = {IEEE International Conference on Computer Vision (ICCV)},
  pages     = {7708--7717},
  year      = {2019}
}

@article{laine2020modular,
  author  = {Laine, Samuli and Hellsten, Janne and Karras, Tero and Seol,
             Yeongho and Lehtinen, Jaakko and Aila, Timo},
  title   = {Modular primitives for high-performance differentiable rendering},
  journal = {ACM Transactions on Graphics (TOG)},
  volume  = {39},
  number  = {6},
  pages   = {1--14},
  year    = {2020}
}

@article{kato2020differentiable,
  author  = {Kato, Hiroharu and Beker, Deniz and Morariu, Mihai and Ando,
             Takahiro and Matsuoka, Toru and Kehl, Wadim and Gaidon, Adrien},
  title   = {Differentiable rendering: A survey},
  journal = {arXiv preprint arXiv:2006.12057},
  year    = {2020}
}

@article{nemhauser1978analysis,
  author  = {Nemhauser, George L. and Wolsey, Laurence A. and Fisher,
             Marshall L.},
  title   = {An analysis of approximations for maximizing submodular set
             functions---{I}},
  journal = {Mathematical Programming},
  volume  = {14},
  number  = {1},
  pages   = {265--294},
  year    = {1978}
}

@article{krause2008near,
  author  = {Krause, Andreas and Singh, Ajit and Guestrin, Carlos},
  title   = {Near-optimal sensor placements in {G}aussian processes: Theory,
             efficient algorithms and empirical studies},
  journal = {Journal of Machine Learning Research},
  volume  = {9},
  pages   = {235--284},
  year    = {2008}
}

@article{haussler1987epsilon,
  author  = {Haussler, David and Welzl, Emo},
  title   = {$\varepsilon$-nets and simplex range queries},
  journal = {Discrete \& Computational Geometry},
  volume  = {2},
  number  = {2},
  pages   = {127--151},
  year    = {1987}
}

@book{santalo2004integral,
  author    = {Santal\'o, Luis A.},
  title     = {Integral Geometry and Geometric Probability},
  edition   = {2},
  publisher = {Cambridge University Press},
  year      = {2004}
}

@inproceedings{vazquez2001viewpoint,
  author    = {V\'azquez, Pere-Pau and Feixas, Miquel and Sbert, Mateu and
               Heidrich, Wolfgang},
  title     = {Viewpoint selection using viewpoint entropy},
  booktitle = {Vision, Modeling, and Visualization (VMV)},
  pages     = {273--280},
  year      = {2001}
}

@article{feixas2009unified,
  author  = {Feixas, Miquel and Sbert, Mateu and Gonz\'alez, Francisco},
  title   = {A unified information-theoretic framework for viewpoint selection
             and mesh saliency},
  journal = {ACM Transactions on Applied Perception},
  volume  = {6},
  number  = {1},
  pages   = {1--23},
  year    = {2009}
}

@article{mamou2009vhacd,
  author  = {Mamou, Khaled and Ghorbel, Faouzi},
  title   = {A simple and efficient approach for {3D} mesh approximate convex
             decomposition},
  journal = {IEEE International Conference on Image Processing (ICIP)},
  pages   = {3501--3504},
  year    = {2009}
}

@article{lien2007approximate,
  author  = {Lien, Jyh-Ming and Amato, Nancy M.},
  title   = {Approximate convex decomposition of polygons},
  journal = {Computational Geometry},
  volume  = {35},
  number  = {1--2},
  pages   = {100--123},
  year    = {2007}
}

@inproceedings{kingma2015adam,
  author    = {Kingma, Diederik P. and Ba, Jimmy},
  title     = {{Adam}: A method for stochastic optimization},
  booktitle = {International Conference on Learning Representations (ICLR)},
  year      = {2015}
}

@article{goodman1983multidimensional,
  author  = {Goodman, Jacob E. and Pollack, Richard},
  title   = {Multidimensional sorting},
  journal = {SIAM Journal on Computing},
  volume  = {12},
  number  = {3},
  pages   = {484--507},
  year    = {1983}
}
